\documentclass[12pt]{article}
\usepackage{graphicx} 
\usepackage{setspace}
\usepackage{amsmath}
\usepackage[letterpaper, margin=1in]{geometry}
\usepackage{tikz}
\usetikzlibrary{snakes}
\usepackage{booktabs}
\usepackage{multirow}
\usetikzlibrary{decorations.pathreplacing,calc}
\usepackage{caption} 
\usepackage{apacite}
\usepackage{natbib}
\usepackage{adjustbox}
\usepackage{longtable}
\usepackage{subcaption}
\usepackage{threeparttable}
\usepackage{multicol}
\usepackage{float}
\usepackage{hyperref}
\usepackage{threeparttable}
\usepackage{array}
\usepackage{amsthm}
\usepackage{mathtools}

\theoremstyle{plain}
\newtheorem{assumption}{Assumption}

\theoremstyle{plain}
\newtheorem{proposition}{Proposition}

\title{A New Testing Method  for Justification Bias Using High-Frequency Data of Health and Employment}

\author{Jiayi Wen \thanks{E-mail: wjyecon@gmail.com. School of Economics and Wang Yanan Institute of Studies in Economics (WISE), Xiamen University, Fujian, China. 
}, Zixi Ye  \thanks{School of Economics, Xiamen University, Fujian, China} , Xuan Zhang \thanks{E-mail: xuanzhang@smu.edu.sg. Singapore Management University, Singapore}  }
\date{February 2024}

 \fontsize{12pt}{18pt}\selectfont

\begin{document}
\maketitle

\begin{abstract}

Justification bias, wherein retirees may report poorer health to rationalize their retirement, poses a major concern to the widely-used measure of self-assessed health in retirement studies. This paper introduces a novel method for testing the presence of this bias in the spirit of regression discontinuity. The underlying idea is that any sudden shift in self-assessed health immediately following retirement is more likely attributable to the bias. Our strategy is facilitated by a unique high-frequency data that offers monthly, in contrast to the typical biennial, information on employment, self-assessed health, and objective health conditions. Across a wider post-retirement time frame, we observe a decline in self-assessed health, potentially stemming from both justification bias and changes in actual health. However, this adverse effect diminishes with shorter intervals, indicating no evidence of such bias. Our method also validates a widely-used indirect testing approach.

   
    
\end{abstract}

\clearpage

\section{Introduction}
\ \ \ \
Understanding the relative importance of health and economic factors in retirement is pivotal for designing policies that incentivize older workers to extend their careers, especially in an ageing society. While detailed diagnostic data allow researchers to probe the effects of specific health symptoms, self-assessed health emerges as an essential metric for gauging the overall influence of health on retirement. It not only offers a comprehensive evaluation of diverse health aspects but also inherently assigns priority to them based on respondents' perceived importance \citep{bound1991}. In addition to its central role in empirical studies examining the influence of health on retirement, self-assessed health is extensively adopted in structural retirement models because of the curse of dimensionality and its holistic nature.\footnote{For regression-based studies that have adopted self-assessed health, see  \cite{disney2006ill}, \cite{bound1999dynamic} as examples. For structural work, see \cite{gustman2005social}, \cite{van2008social}, \cite{french2011effects},  among others. }.

Despite these merits, a critical concern about self-assessed health is reporting bias. Justification bias emerges as a systematic reporting bias when retirees employ health as a more socially acceptable justification for retirement rather than providing authentic reasons. By inflating the effects of health on retirement, it has long been accused of an identification threat  in the literature \citep{anderson1985retirement, dwyer1999, disney2006ill, bound2010health, blundell2023}.  Existing studies commonly test its presence by comparing instrumental variable (IV) estimates with the raw retirement effect of self-assessed health, assuming that objective health, used as the instrument, is exogenous \citep{stern1989measuring,bound1991,blundell2023}. While studies based on this indirect approach generally do not uncover strong evidence of justification bias, other research that explores health variations by retirement status find such bias is nonignorable and large \citep{lindeboom2009health},\citep{black2017justification}.

In this paper, we introduce a novel method for testing the presence of justification bias.  Our key idea hinges on that sudden changes in self-assessed health immediately after retirement are more likely to be attributed to reporting bias rather than changes in actual health. We leverage a distinctive dataset, Singapore Life Panel (SLP), which provides information at the monthly level on individuals' employment, self-assessed health, and objective health conditions. Unlike standard data sets for older populations, usually surveyed every two years, this high-frequency data set empowers us to meticulously scrutinize abrupt changes in health, both subjective and objective, around retirement.

We begin by isolating the causal effect of retirement on self-assessed health from the observed health and retirement correlation, exploiting the statutory retirement age as IV, following a number of recent literature exploring the health implications of retirement.\footnote{See \cite{eibich2015} and \cite{Rose2020} for examples.}. We find an adverse effect of retirement on self-assessed health. Nevertheless, while this outcome aligns with the potential existence of justification bias, it is also plausible that it is influenced by shifts in actual health status.

To further disentangle reporting bias from changes in actual health, we apply our new approach under two sets of identification assumptions. The first assumption posits that while retirement may influence health transitions through lifestyle changes, the impact on the \textit{level} of actual health should manifest over time without abrupt changes. Detection of sudden shifts in the level of self-assessed health within a sufficiently narrow timeframe is thus attributable to reporting bias. In a broader window around the statutory retirement age, we observe that retirement leads to a decline in self-assessed health. However, as the window narrows, the estimated effect diminishes, and with sufficiently narrow windows, we cannot reject the null hypothesis of zero retirement effect on self-assessed health. Consequently, we find no evidence of justification bias.

The first identification assumption focuses only on the variation in self-assessed health. With monthly information on various objective health conditions, we can further relax this assumption. Our second identification assumption only considers a sudden shift in self-assessed health relative to objective health as reporting bias. Intuitively, the variation in objective health can serve as the ``control group''. Therefore, our identification resembles the ``difference-in-differences'' framework, capturing the difference in subjective versus objective health before and after retirement. Similarly, we find that while self-assessed health experiences a larger decline compared to objective health in a broader post-retirement window, the gap diminishes as the time frame narrows, becoming statistically insignificant under a sufficiently small window surrounding the retirement age.

As we discover a null effect of retirement within a shorter period, it is imperative to ensure that this insignificance is not a result of limited statistical power under smaller sample sizes. In contrast to the diminishing coefficients, corresponding standard errors remain highly stable as the time frame narrows. The estimates even appear more precise due to stronger first stages.

Additional exercises test justification bias based on the widely-used indirect approach. The raw effect of self-assessed health on retirement is notably smaller than the IV estimate, indicating that attenuation bias from measurement errors dominates justification bias. Therefore, our new method validates this common indirect approach for testing justification bias. We also apply our new method to the Health and Retirement Study (HRS) data set from the US, a typical biennial survey. We find that there is no enough power to make credible conclusions. While this comparison highlights the advantage of high-frequency data, our method has the potential applicability to standard data sets as their waves and sample size accumulate.  Additionally, we examine the immediate effects of retirement on various objective health conditions, finding no significant results, which further validates our first identification assumption that actual health should not change suddenly.


This paper introduces an innovative method to address a longstanding challenge in the study of the relationship between health and retirement.  Early investigations into retirement often reveal more significant effects of health compared to economic variables when using subjective measures, with justification bias frequently cited as a prominent explanation \citep{anderson1985retirement,bazzoli1985early}. Assuming the exogeneity of objective health, \cite{stern1989measuring} proposes its use as an instrument for subjective health. The IV estimate appears to be close to, yet slightly larger than, the OLS estimate, ruling out justification bias. Similar results are found by \cite{bound1991} and \cite{dwyer1999}. Using comprehensive objective measures, \cite{blundell2023} similarly conclude that ``measurement error and justification bias are not important sources of bias, or at least that the two sources of bias offset one another.''

In contrast, another strand of studies explores variations in health across employment status to test for justification bias, often finding nontrivial bias.  For instance, \cite{lindeboom2009health} identify justification bias by examining the impact of retirement on additional variation in self-assessed health, after accounting for objective health.\footnote{Similar strategy is adopted by\cite{kerkhofs1995subjective} and \cite{gannon2009influence}.} \cite{black2017justification} exploits a unique feature in an Australian longitudinal survey that asks an identical question about disability twice. They then test justification bias by variations in the difference of these two questions and variations of employment. 

Our approach distinguishes itself from previous studies through the integration of methodological advancements in regression discontinuity and the utilization of unique high-frequency data. The identification assumption that health changes should not occur abruptly following retirement becomes notably more plausible within a narrow time frame.

Additionally, this paper adds to broad discussions on the reliability of self-assessed health. Beyond justification bias, this literature concerns about other forms of reporting bias related to reference \citep{lindeboom2004cut}, threshold \citep{kapteyn2007vignettes}, or survey designs \citep{crossley2002reliability,clarke2006self,lumsdaine2013survey}. This paper helps address concerns regarding the credibility of self-assessed health as a widely adopted measure.

Furthermore, our work adds value to the literature exploring the health implications of retirement. This literature are generally inconsistent between findings on self-assessed health and on objective health \citep{garrouste2022there}. While discernible effect is observed on self-assessed health, the impact on objective health conditions is less clear. Despite the identification of causal effects on self-assessed health, reporting bias cannot be ruled out \citep{eibich2015, Rose2020}.\footnote{\cite{Rose2020} uses census data and also finds an immediate improvement in subjective health after the statutory retirement age. With a large sample size, the estimates are tight and the pattern is clear. However, objective health measures are unavailable in the census data.} Taking into account the dynamics of health, our approach offers a possibility to disentangle the endogeneity of self-assessed health resulting from reporting bias and changes in actual health.


The next section introduces the framework for the new testing method, followed by Section 3 that introduces the empirical strategy, background and data. Section 4 presents the main results based on two alternative identification assumptions. Additional relevant results, such as the results based on the conventional method, are provided in Section 6. Section 7 shows various robustness checks. Section 8 concludes.

\section{Conceptual Framework}\label{section:Conceptual Framework and Econometric Methods}
\ \ \ \
In this section, we first review the standard econometric framework of justification bias and a popular indirect testing method. Then we introduce our new method along with two alternative identification assumptions.

\subsection{Justification Bias}
\ \ \ \
Consider a simple statistical model of retirement:
\begin{equation}
R=\theta_0+\theta_{H}H+\theta_{W}^{'}W+e
\label{reg:retirement}
\end{equation}
 where $R$ is an indicator of retirement and $H$ is individual's actual health status.\footnote{With a binary indicator, this equation can be conceptualized as a linear probability model. $H$ can represent various health conditions and need not be a scalar. We conceptualize it as a single variable. Both simplifications are for ease of exposition, without a loss of generality.} We assume that larger values of $H$ correspond with worse health, so that the coefficient $\theta_H$ is expected to be positive. $W$ are other observable determinants of retirement, such as age and gender, whereas $e$ captures residual determinants of employment that are unobservable. To assess the overall impact of health, we would like to have a summary measure of the actual health $H$. Typical candidates include subjective measures that are self-assessed by the respondent, or sets of detailed health conditions that are more objective.

As objective measures usually consist of only a subset of conditions, they may miss important dimensions of health that are related to worker's retirement decision \citep{bound1991,blundell2023}. Instead, self-assessed health can serve as a comprehensive measure whereby individuals tend to weigh poor health conditions by the extent of disruption to them. Meanwhile, self-assessed health has also been widely used in studies on individuals' labor market outcomes beyond retirement.\footnote{See \cite{o2015health} for examples.}
 
A major concern for subjective health measures is the reporting bias.\footnote{We will use subjective health and self-assessed health interchangeably in this paper.} Suppose that the measure of subjective health $H^s$ combines the true health status $H$ with biases, of which two sources are typically considered: the classical measurement error $\nu$ and the state-dependent justification bias $\epsilon$.
\begin{equation}
    H^s=H+v+\epsilon
\end{equation}
Even though the measurement error is random, estimates of the effect $\theta_H$ will suffer from the attenuation bias and be biased towards zero. On the other hand, the justification bias indicates that, as retirees are more likely to use health as the reason to justify their retirement, the reporting bias $v$ may be positively correlated with the residuals in the retirement equation $e$, leading to exaggerated effect of health. 

To see these two biases, consider the OLS estimator $\hat{\theta}_H^{OLS}$ for Equation \eqref{reg:retirement}, using $H^s$ to proxy $H$. It can be shown to have the following asymptotic limit:
\begin{equation}
  \text{plim }  \hat{\theta}_H^{OLS} = \frac{Cov(R,H^s)}{Var(H^s)}
    =\frac{\theta_H Var(H)+Cov(e,\epsilon)}{Var(H)+Var(v)+Var(\epsilon)} 
    \label{reg:OLS}
\end{equation}
The object to which the OLS estimator converges may be greater or smaller than the parameter of interest $\theta_H$, depending on the justification bias $Cov(e, \epsilon)$, which is positively correlated, as well as the attenuation bias caused by larger variance in the denominator.\footnote{For results in Equation \eqref{reg:OLS}, the actual health $H$ is typically assumed to be orthogonal to the bias $\nu$ and $\epsilon$, as well as the unobserved retirement determinants $e$ after controlling for the observables $W$. We also apply the Frisch–Waugh–Lovell theorem to abstract from the observables $W$.}  \cite{stern1989measuring} proposes the usage of objective health conditions as instruments for the subjective health to recover the effect $\theta_H$. A number of subsequent studies, such as \cite{bound1991}, \cite{dwyer1999} and \cite{blundell2023},  infer the dominant type of bias by comparing the magnitude of OLS and IV estimates.

\subsection{The New Testing Method}
\ \ \ \
For our testing method, we start by modelling the conditional expectation of actual health status $H$ as a function of age $X$. General knowledge suggests that health status should exhibit a certain smoothness as age naturally progresses, while retirement may change the transition of health and the age function needs not to be smooth anymore.  Therefore, we allow the functional form to differ by retirement status $R$. To focus on our main idea, retirement $R$ is taken as exogenous in this framework. Therefore, the current framework is better to be taken as focusing on the impact of retirement on health, abstracting from the reverse effect of health on retirement.\footnote{Our solution to the issue of simultaneous causality is referred to the Empirical Strategy section.} 

Retirement may improve individuals' health by changing the lifestyle, such as having more leisure to do exercise or see doctors. Alternatively, it may fasten the pace of health deterioration by lowering consumption quality with reduced financial resources. The \textit{transition} of health thus may well be influenced by retirement. The \textit{level} of actual health, however, is unlikely to subject to abrupt change. This intuition is supported by the the fact that health is widely conceptualized as a stock following the influential work by \cite{grossman1972concept}, whereas retirement influences the investment on it. This observation motivates our first identification assumption: while the conditional expectation function of actual health is allowed to be non-smooth at ages when the individual retires $X=x^*$, it should remain as being continuous.\footnote{We take retirement as an absorbing stage so that the retirement age $x^*$ is clearly defined. The retirement may be a gradual process in reality for some people, but this fact should not weaken, if not strengthen, the assumption of continuity.}  
\begin{assumption}
 The conditional expectation of actual health H is a smooth function of age X. Taking retirement into account, the function is allowed to be non-smooth at the retirement age $ X=x^{*} $, but it remains being continuous. Specifically,
     \begin{align*}
  E[H | X, R] \equiv  & \: f(X,R) =(1-R)\cdot f_0(X) +  R\cdot f_1(X)  
  \shortintertext{where} 
   &  \lim_{\delta \rightarrow 0 } f_0(x^{*}-\delta) =     \lim_{\delta \rightarrow 0 } f_1(x^{*}+\delta)  
    \end{align*}
  and $f_0$,  $f_1$ are smooth functions of $x$ on the supports $\mathcal{X}$.
\end{assumption}
Now consider the self-assessed health $H^s= H + \nu + \epsilon $. The measurement error $\nu$ is random so that it is mean-independent with $E[\nu|X, R]=0$. By definition, justification bias indicates that the bias $\epsilon$ is state-dependent on the retirement status, with:
\begin{align*}
    \lambda \equiv E(\epsilon|X, R=1) > E(\epsilon|X, R=0) \equiv 0
\end{align*}
The conditional expectation of self-assessed health $H^s$ is thus given by:
 \begin{align*}
E[ H^s | X, R] = & E[H +   \nu + \epsilon | X, R]  \\
     = & (1-R)\cdot f_0(X) +  R\cdot f_1(X) + R\cdot E(\epsilon| X, R=1) \\
     = & (1-R)\cdot f_0(X) +  R\cdot \tilde{f_1}(X) 
      \shortintertext{ where }
      \tilde{f_1}(X) =&  f_1(X) + \lambda 
\end{align*}
 Therefore, under Assumption 1, we can obtain the following proposition regarding the identification of justification bias defined by $\lambda$.\footnote{Strictly, we assume away heterogeneity in the justification bias at the individual level $\epsilon$, so that $\epsilon=\lambda$.} 

\begin{proposition}
Suppose that Assumption 1 holds. Justification bias $\lambda$ is identified by the abrupt change in conditional expectation of self-assessed health around the retirement age $x^*$.
\end{proposition}

\begin{proof}
 \begin{align*}
   &     \lim_{\delta \rightarrow 0 } E[H^s | X=x^{*}+\delta, R=1]  -  \lim_{\delta \rightarrow 0 } E[ H^s | X=x^{*}-\delta, R=0]   \\  =& \lim_{\delta \rightarrow 0 } \tilde{f_1}(x^{*}+\delta) -  \lim_{\delta \rightarrow 0 } f_0(x^{*}-\delta)  \\
 = & \lim_{\delta \rightarrow 0 } f_1(x^{*}+\delta) +  \lim_{\delta \rightarrow 0 }  E(\epsilon| X= x^{*} + \delta, R=1) -  \lim_{\delta \rightarrow 0 } f_0(x^{*}-\delta)  \\
 = & \lambda  &&\qedhere
 \end{align*} 
 \end{proof}

Assumption 1 implies that any sudden change in self-assessed health after retirement is more likely to be attributed to reporting bias rather than the alternation in real health. While it is plausible, we can further relax it if we also observe health conditions $H^o$ that are more objective. Under our alternative Assumption 2, only the sudden change in self-assessed health relative to objective health conditions is attributed to justification bias.

Formally, suppose the actual health $H$ can be decomposed into the observed conditions $H^o$ and the remaining unobserved conditions $q$, such that $H=H^o+q$. Instead of requiring $H$ to be continuous at the retirement age, Assumption 2 requires that the unobserved conditions $q$ is a continuous function of age. 

\begin{assumption}
     The unobserved health conditions $q$ are defined as the difference of the actual health $H$ and the observed health conditions $H^o$. Its conditional expectation function with respect to age is  continuous over the support $\mathcal{X} $, including the retirement age $ X=x^{*} $. Specifically,
     \begin{align*}
  E[q | X, R] \equiv  & \, g(X,R) =(1-R)\cdot g_0(X) +  R\cdot g_1(X)  
  \shortintertext{where}
   &  \lim_{\delta \rightarrow 0 } g_0(x^{*}-\delta) =     \lim_{\delta \rightarrow 0 } g_1(x^{*}+\delta) 
 \end{align*}
  and $g_0$,  $g_1$ are smooth functions of $x$ on the supports  
\end{assumption}
Under Assumption 2, we can obtain the second proposition regarding the identification of justfication bias $\lambda$:
\begin{proposition}
    Suppose that Assumption 2 holds. Justification bias $\lambda$ is identified by the abrupt change in conditional expectation of self-assessed health relative to objective health around the retirement age $x^*$.
\end{proposition}
\begin{proof}
   \begin{align*}
   &    \big \{ \lim_{\delta \rightarrow 0 } E[H^s | X=x^{*}+\delta, R=1]  -  \lim_{\delta \rightarrow 0 } E[ H^s | X=x^{*}-\delta, R=0]  \big \}  \\ 
      & \quad  -  \big \{ \lim_{\delta \rightarrow 0 } E[H^o | X=x^{*}+\delta, R=1]  -  \lim_{\delta \rightarrow 0 } E[ H^o | X=x^{*}-\delta, R=0]  \big \}  \\  
 =  &  \big \{    \lim_{\delta \rightarrow 0 } E[H^o+q+\nu+\epsilon | X=x^{*}+\delta, R=1]  \\
 & \quad  -  \lim_{\delta \rightarrow 0 } E[ H^o+q+\nu+\epsilon | X=x^{*}-\delta, R=0]   \big \}  \\ 
    &  \quad   -  \big \{ \lim_{\delta \rightarrow 0 } E[H^o | X=x^{*}+\delta, R=1]  -  \lim_{\delta \rightarrow 0 } E[ H^o | X=x^{*}-\delta, R=0]  \big \}  \\ 
  =  &  \lim_{\delta \rightarrow 0 } E[q | X=x^{*}+\delta, R=1]  -  \lim_{\delta \rightarrow 0 } E[ q | X=x^{*}-\delta, R=0]  \\
  & \quad  + \lim_{\delta \rightarrow 0 } E[ \epsilon | X=x^{*}+\delta, R=1]  -  \lim_{\delta \rightarrow 0 } E[ \epsilon | X=x^{*}-\delta, R=0]  \\
 = & \lim_{\delta \rightarrow 0 } g_1(x^{*}+\delta)   -  \lim_{\delta \rightarrow 0 } g_0(x^{*}-\delta)    +  \lim_{\delta \rightarrow 0 }  E(\epsilon| X= x^{*} + \delta, R=1) \\
 = & \lambda   &&\qedhere
 \end{align*}  
\end{proof}

A remaining challenge is that the actual retirement age $x^*$ is notoriously difficult to pinpoint. Proposition 3 suggests that the above identification strategies are also applicable to a small time frame around the statutory retirement age instead of the actual retirement. The statutory retirement age is regulated by pension laws, at which workers are more likely to retire with extra financial incentives.\footnote{The background of statutory retirement ages in Singapore is provided in the next section.} To stay focused, please referred to Appendix A for detailed discussions and the proof of Proposition 3.

\begin{proposition}
The abrupt change in self-assessed health around the actual retirement will translate into an abrupt change within a short interval around the statutory retirement age, as long as the mass of individuals who are induced to retired by the statutory retirement age is positive.
\end{proposition}


\section{Background, Empirical Strategy, and Data}
\subsection{Background}
\ \ \ \
We apply our method using the SLP data. In Singapore, the Retirement and Re-employment Act (RRA) was introduced in 2012 to deal with challenges posed by an aging workforce, which encourages employers to retain and rehire older workers. According to RRA, employers are not allowed to force their employees to retire before the MRA. Employers must also provide the option of continued employment, known as re-employment, to eligible workers who reach the MRA, until they reach the maximum re-employment age (REA). If a worker is eligible for re-employment but the employer can not offer a suitable position, the employer must assist the worker in finding another job or provide a lump-sum employment assistance payment. It is noteworthy that even though the government urges employers to extend re-employment contracts to eligible employees, the associated penalty is minor.

Meanwhile, the cornerstone of Singapore's pension system is the Central Provident Fund (CPF), which was established in 1955. This fund operates as a compulsory defined-contribution scheme. Both employers and workers make monthly contributions, which are subject to an earning limit. Upon attaining the age of 55, individuals become eligible to make withdrawals from their CPF, provided that they maintain a minimum balance. After reaching the pension eligibility age (PEA), individuals are allowed to receive monthly pension benefits.

\begin{figure}[H]
\centering
\begin{minipage}{0.9\textwidth}
\begin{tikzpicture}[snake=zigzag, line before snake = 10mm, line after snake = 10mm]

    \draw (-2,0) -- (13,0);

    \foreach \x in {1,5.5,9.5}
      \draw (\x cm,3pt) -- (\x cm,-3pt);

    \node[align=center] at (1,-1) {\textbf{MRA}};
    \node[align=center] at (5.5,1.75) {\textbf{PEA}};
    \node[align=center] at (9.5,-1) {\textbf{REA}};

  \node[align=center] at (1.5, -2) {{\footnotesize Employees cannot be forced to retire}\\ {\footnotesize prior to this age }  };
    \node[align=center] at (5.5, 1) {{\footnotesize Entitled to receive monthly pension benefits } };
    \node[align=center] at (9.5, -2) { {\footnotesize Employers should re-employ} \\ {\footnotesize the elder until this age} };

 \end{tikzpicture}
    \caption{Statutory Retirement Ages in Singapore}
    \vspace{0.2cm}
   
         \scriptsize {\textit{Notes:} This graph lists three statutory retirement ages in Singapore. MRA: Minimum Retirement Age; PEA: Pension Eligibility Age; REA: Maximum Re-Employment Age.} 
          \end{minipage}
   
\end{figure}

Therefore, there are three statutory retirement ages at which individuals are more likely to exit the labor market: the MRA, the PEA, and the REA. As prescribed by law, these ages are exogenous to individual's private characteristics related to health and employment, such as their socioeconomic status. Hence, these ages are plausibly good candidates of instrumental variables, once the smooth effects of aging on retirement are controlled for.

\begin{table}[H]
  \centering
  \caption{ Statutory Retirement Ages of Different Cohorts}
  \begin{minipage}{0.99\textwidth}
{  \footnotesize
  \tabcolsep=0.03cm
    \begin{tabular}{p{0.22\textwidth}>{\centering}p{0.08\textwidth}|>{\centering}p{0.27\textwidth}>{\centering}p{0.08\textwidth}|>{\centering}p{0.25\textwidth}>{\centering\arraybackslash}p{0.08\textwidth}}
    \hline
    \hline
    \multicolumn{2}{c|}{\textbf{MRA}} & \multicolumn{2}{c|}{\textbf{PEA}} & \multicolumn{2}{c}{\textbf{REA}} \\  
    \multicolumn{1}{c}{Cohort} &  \multicolumn{1}{c|}{Age}& \multicolumn{1}{c}{Cohort} &  \multicolumn{1}{c|}{Age} &  \multicolumn{1}{c}{Cohort} &  \multicolumn{1}{c}{Age} \\
    \hline
    \multicolumn{1}{c}{Birth $<$ 1960.7}  & 62    & 1939.1 $\leq$ Birth $<$ 1950.1 & 62    & Birth $<$ 1952.7 & 65 \\
     \multicolumn{1}{c}{Birth $\geq$ 1960.7} & 63    & 1950.1 $\leq$ Birth $<$ 1952.1 & 63    & 1952.7 $\leq$ Birth $<$ 1955.7 & 67 \\ 
          &       & 1952.1 $\leq$ Birth $<$ 1954.1 & 64    & Birth $\geq$ 1955.7 & 68 \\
          &       & Birth $\geq$ 1954.1 & 65    &       &  \\
    \hline
    \end{tabular}%
  \label{table:bg1}  }

\scriptsize \textit{Notes:} This table presents the cohort-specific MRA, PEA, and REA. Birth refers to the birth year and month. MRA: Minimum Retirement Age; PEA: Pension Eligibility Age; REA: Maximum Re-Employment Age. 
  \end{minipage}

      \end{table}%

Note that these statutory retirement ages are subject to policy changes. The MRA was originally set as 62 since 1999 until June 2022, after which it was increased to 63 in July 2022. The REA stood at 65 initially and elevated to 67 since July 2017, after which it underwent further revisions. On July 1, 2022 the REA became 68, and it will rise to 70 by 2030. The PEA was initially fixed at 60 but was raised to 62 in 1999 to align with the MRA. It was subsequently increased to 63 in 2012, 64 in 2015, and 65 in 2018. Table \ref{table:bg1} summarizes the statutory retirement ages by different birth cohorts, and we calculate individual-specific statutory ages according to this table.

\subsection{Empirical Strategy} 
\ \ \ \
By taking retirement as exogenous, the previous conceptual framework focuses on the effect of retirement on health and abstracts from the reverse effect of health on retirement. Empirically, the observed variation in health before and after retirement merely reflects the correlation. 
Therefore, our empirical strategy involves two steps. Firstly, to disentangle the influence of retirement on health from the reverse effect of health on retirement, we use the statutory retirement ages as instruments, following the burgeoning literature that aims to identify the causal effects of retirement on health, cognition, and healthcare utilization. While the first step allows the estimate to capture the effect of retirement on self-assessed health, it encompasses both genuine health changes and justification bias. Secondly, to differentiate between these two, we apply the two alternative identification assumptions introduced in the conceptual framework, by narrowing the time frame around the statutory retirement age. The second step is also facilitated by the unique high-frequency data on employment and various health measures.

For the first step, to identify the causal effect of retirement, two main methods are typically used in the literature: the regression discontinuity design (RD) and the fixed-effect instrumental variable approach (FE-IV). Both methods exploit the discontinuous jump of retirement probability at the statutory retirement ages after controlling for the natural and smooth impacts of aging. The fuzzy RD can also be considered as an IV approach in essence. The difference is that RD approach allows different parametric forms of the age function before and after the cutoff, whereas the FE-IV approach, while usually controls for a smooth age function, incorporates the individual fixed effects to absorb any time-invariant omitted variables.

We follow both approaches by adding individual fixed effects and controlling for flexible functions of age. We use the binary variable that indicates whether the statutory retirement age has been reached as the instrument. While in theory the MRA, the PEA, and the REA can shift the retirement likelihood, we find that the PEA and the REA have very small effects empirically, so that the first stage is weak. It is also a cleaner setting to focus on the change in health around a single cutoff. Therefore, we focus on the MRA in our main analysis. 

The following two equations are estimated in the two-stage least squares estimation:
\begin{align}
R_{it}=&\alpha+\rho \cdot \mathbf{1}\{X_{it}>m^*_{i} \}+g_R(X_{it})+\phi_R^{'}W_{it}+\theta_i+e_{it}
 \label{eqIV1} \\
H_{it}=&\beta+\gamma \hat{R}_{it}+g_H(X_{it})+ \phi_H^{'}W_{it}+\sigma_{i}+\epsilon_{it}
 \label{eqIV2}
\end{align}
where $H_{it}$ is the health outcome, measured by self-assessed health or objective health conditions, and $R_{it}$ is the indicator variable of retirement. $m^*_i$ denotes the MRA, which may vary across individuals depending on their birth cohorts, and $\mathbf{1}\{X_{it}>m^*_{i} \}$ is the indicator of reaching it. $X_{it}$ denotes the age centered around the MRA, so that $\rho$ and $\gamma$ capture the abrupt changes at the MRA. $g_R()$ and $g_H()$ are functions of age that captures the natural progress of retirement and health deterioration.\footnote{In our main analysis, we control for a smooth function of age with second order polynomials, akin to the FE-IV approach. We also allow the age function to differ before and after the cutoff in the robustness check.} 
$W_{it}$ is a vector of additional control variables. $\theta_i$ and $\sigma_i$ represent individual fixed effects, and $e_{it}$ and $\epsilon_{it}$ are idiosyncratic errors. What sets our approach apart from a standard fuzzy RD design is the incorporation of individual fixed effects. This directs our focus towards the variation over time, being more coherent with our conceptual framework that focuses on the within-individual variation. 

The above first step only identifies the effect of retirement on self-assessed health. Our second step further separates the reporting bias and the change in real health by examining any abrupt change in health after retirement. To achieve this goal, we repeat estimating Equation \ref{eqIV1} and \ref{eqIV2} using sub-samples under various intervals around the MRA. As revealed by Proposition 3, any significant effect of retirement on self-assessed health in a sufficiently small interval around the MRA is indicative of the existence of justification bias.

Although the RD or the FE-IV design should, in theory, identify the effect at the cutoff. However, standard survey data suffer from two issues to apply our method. First, they are sensitive to parametric functional form approximations. This issue is exacerbated by the unavoidable interpolation or extrapolation due to the sparsity of biennial survey designs. Second, even with correctly-specified functional forms, they may lack sufficient statistical power to make meaningful conclusions. The high-frequency data allow us to obtain enough observations within a much narrower time frame. This not only helps mitigate the bias stemming from potential misspecification in the functional form but also enhances the statistical power of our tests.

\subsection{Data, Sample, and Measures} \label{empirical}
\ \ \ \
The data set we exploit is the SLP, a monthly panel survey conducted since August 2015.\footnote{More information about SLP can be found in \cite{vaithianathan2021}} We use data until January 2020 to avoid the disruption in employment and health status due to COVID-19, which includes 54 waves in total. The design is similar to the HRS in the US, but different sets of questions are asked at four different frequencies: monthly, quarterly, annually, and one-off. The data set includes a wide range of information about economic, social, and health conditions of older adults in Singapore. Eligible respondents are Singaporeans (citizens and permanent residents) aged 50 and over and their spouses, with a monthly sample size ranging from 7,000 to 9,000 individuals. The data set is unique in that it collects information every month but maintains similar survey designs as counterparts in many other countries which are typically biennial, such as the HRS in the US, the English Longitudinal Study of Ageing (ELSA), the Survey of Health, Ageing and Retirement (SHARE) in Europe, and the China Health and Retirement Longitudinal Study (CHARLS). In particular, in additional to labor market outcomes, SLP incorporates both subjective health measures and various objective health conditions at the monthly level. This rich information is helpful for implementing test strategies under our identification assumptions.

For our analyses, we restrict observations to individuals between age 55 and 75. Observations reporting ``homemaker'', ``student'', or ``others'' are excluded from analyses. Meanwhile, the RRA is applicable only to individuals who work as employees but not to those self-employed. Therefore, observations reporting ``self-employed'' are also excluded. Additionally, in line with the conceptual framework, we exclude individuals with no variation in labor force participation throughout the study period. Individuals reporting ``working for pay'', ``unemployed and looking for work'', ``temporarily laid off'', or ``on sick or other leave'' are defined as \textit{in the labor force}. 

For the measures, an individual's age used in our analyses is measured in months. We also center it around the MRA, so that the age at which an individual reaches the MRA is normalized as 0. 
We define retirement based on an individual's employment. An individual is considered as retired if she did not report ``working for pay'' in a multiple-choice questions about the current job status. We drop observations with more than one answers to maintain a cleaner sample.

Following the literature, health variables can be classified into two broad categories: subjective and objective health. For subjective health, it is measured by the self-assessed health with a scale from 1 to 5, where 1 represents ``excellent''; 2 ``very good''; 3 ``good''; 4 ``fair''; and 5 ``poor''. On top of this raw measure, we also construct a binary variable which indicates whether the subjective health is poor. If a respondent reported ``fair'' or ``poor'', this binary variable is assigned a value of 1, and 0 otherwise. As for objective health, it is based on individuals’ self-reported objective health conditions. Specifically, respondents were asked to answer the question: ``In last month, did a doctor tell you that you have any of the following conditions?'', where the conditions including hypertension, diabetes, cancer, heart problems, stroke, arthritis, and psychiatric problems. The number of bad health conditions is also summed up as a single index to act as a summary of individuals' objective health status.\footnote{Note that while these specific health conditions are self-reported rather than obtained from clinical diagnoses, they are not involved in a comparison with other individuals. Therefore, these measures are not subject to the biases caused by varying reference points, which may be different across individuals and over time.} In the robustness checks, we consider additional time-varying control variables: survey year dummies, health insurance expenditure, and marital status.\footnote{The questionnaire includes a section on respondents' health spending. The module asks the amount spent by respondents and their spouses on various health-related items in the previous month, among which health insurance is one of the items. To avoid measurement errors, we only consider a binary status of whether having any positive expenditure on health insurance.}  Table \ref{tab:Summary_Statistics} reports summary statistics of the main variables used in our analyses.

\begin{table}[H]
  \centering
  \footnotesize
  \caption{Summary Statistics of Key Variables}
 \tabcolsep=0.12cm
\begin{minipage}{0.99\textwidth}
   \footnotesize
   \begin{tabular}{p{0.33\textwidth}>{\centering}p{0.12\textwidth}>{\centering}p{0.15\textwidth}>{\centering}p{0.15\textwidth}>{\centering}p{0.08\textwidth}>{\centering\arraybackslash}p{0.08\textwidth}}
    \hline
    \hline
    \textbf{Variables} & \textbf{Obs.} & \textbf{Mean} & \textbf{Std.} & \textbf{Min} & \textbf{Max} \\
    \hline
    Retirement & 93,818 & 0.398 & 0.489 & 0     & 1 \\
     $\mathbf{1}$[Age$>$MRA]   & 93,818 & 0.603 & 0.489 & 0     & 1 \\
    Self-assessed Health & 93,818 & 3.262 & 0.869 & 1     & 5 \\
    Poor Health Indicator & 93,818 & 0.389 & 0.488 & 0     & 1 \\
    Objective Health Index & 93,818 & 0.285 & 0.671 & 0     & 7 \\
   \quad  Hypertension & 93,818 & 0.119 & 0.324 & 0     & 1 \\
  \quad  Diabetes & 93,818 & 0.080 & 0.271 & 0     & 1 \\
   \quad    Cancer & 93,818 & 0.012 & 0.110 & 0     & 1 \\
   \quad    Heart Problems & 93,818 & 0.031 & 0.172 & 0     & 1 \\
   \quad    Stroke & 93,818 & 0.006 & 0.076 & 0     & 1 \\
    \quad   Arthritis & 93,818 & 0.031 & 0.174 & 0     & 1 \\
    \quad   Psychiatric Problems & 93,818 & 0.006 & 0.078 & 0     & 1 \\
    Age  & 93,818 & 1.509 & 4.905 &  -8  & 13 \\
    Health Insurance & 93,818 & 0.253 & 0.435 & 0     & 1 \\
    Marriage Status & 93,802 & 0.790 & 0.408 & 0     & 1 \\
    \hline
    \end{tabular}%
  \label{tab:Summary_Statistics}

 \scriptsize \textit{Notes:} This table shows the mean, standard deviation, minimum and maximum value of each variable used for analyses.   For Self-assessed Health, a larger value indicates worse status. 
  Objective Health Index counts the number of chronic conditions a respondent had in the previous month. Here, we report age on yearly basis, and it is centralized around the MRA, the same as the $X_{it}$ in Eqs \eqref{eqIV1} and Eqs \eqref{eqIV2}.

    \end{minipage}
\end{table}%


\section{Main Results}
\ \ \ \ 
In what follows, we first show the effect of retirement on self-assessed health based on the FE-IV strategy. Then we apply our identification strategy to separate the reporting bias from the changes in actual health.

\subsection{The Effect of Retirement on Self-Assessed Health}
\ \ \ \ 
To isolate the effect of retirement on health from their correlation, we first adopt the FE-IV approach, following the extensive literature on the health effect of retirement. The basic intuition is that after controlling for the smooth effect of aging on retirement, the retirement likelihood jumps once individuals reach the statutory retirement age due to policy incentives. Whether an individual has reached the statutory retirement age thus can serve as an instrumental variable under the conditional instrumental relevance and the conditional exogeneity assumptions. 

Figure \ref{fig:MRA, Retirement and Self-Assessed Health} presents changes in the retirement probability and self-assessed health before and after reaching the MRA. We can observe a noticeable increase in the retirement probability once individuals reach the MRA. The graph on the right of Figure \ref{fig:MRA, Retirement and Self-Assessed Health} displays changes in self-assessed health before and after the MRA. Although there seems to be no discernible jump, the transition of health, as reflected by the slope, appears to be upward-sloping and steeper after the MRA, suggesting a deterioration of self-assessed health.

\begin{figure}[H]
\begin{minipage}{0.5\linewidth}
  \centering
\includegraphics[width=1\linewidth]{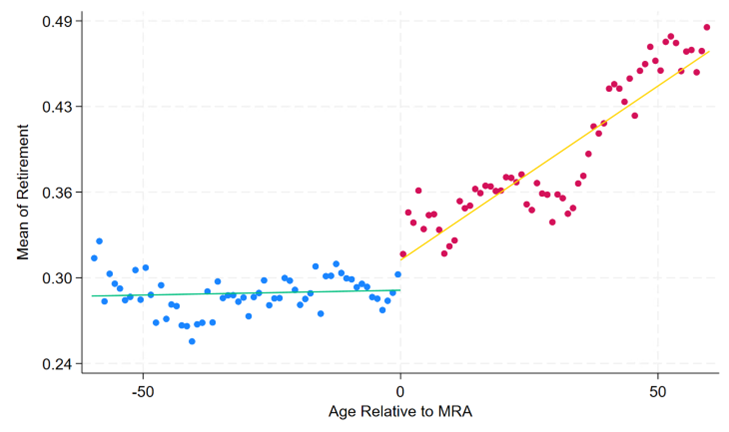}
\end{minipage}%
\begin{minipage}{0.5\linewidth}
  \centering
\includegraphics[width=1\linewidth]{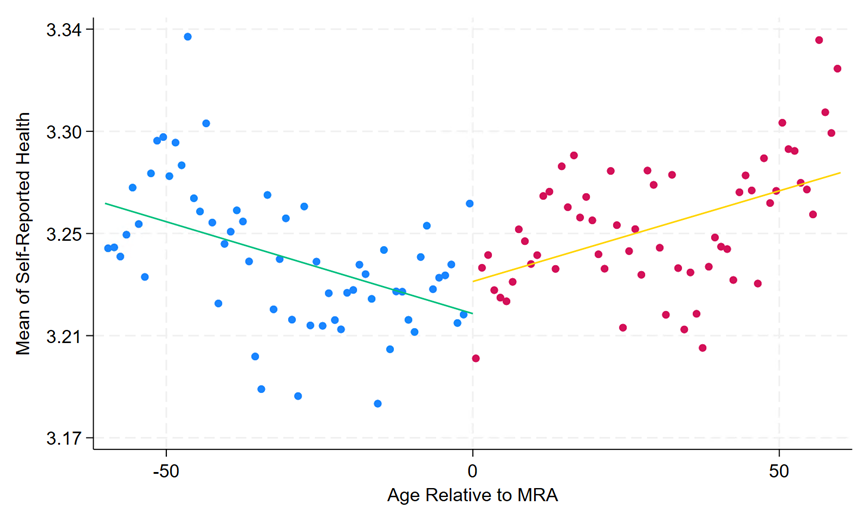}
\end{minipage}
\caption{MRA, Retirement and Self-Assessed Health}
\label{fig:MRA, Retirement and Self-Assessed Health}

\vspace{0.2cm}
    \scriptsize \textit{Notes:} Each point represents the average of retirement and self-assessed health in each age. Age,  measured in month, is centered around the Minimum Retirement Age (MRA), which is calculated for each individual based on their cohorts.
\end{figure}

Estimates presented in Table \ref{tab:SAH1} show that reaching the MRA increases the retirement probability by 4.62 percentage points, whereas retirement worsens self-assessed health by approximately one standard deviation. This detrimental effect of retirement on health is in line with the existence of justification bias. However, this is not the only explanation, as it might also capture the effect of retirement on actual health.\footnote{In particular, it may capture the bias caused by any misspecification of the age functions in the health equation.} To separate the effect of retirement on actual health and the reporting bias induced by retirement, we proceed with the analysis invoking our first identification strategy. Nevertheless, this baseline result, without restrictions on time frames around retirement, serves as the benchmark for our following sub-sample analyses.

\begin{table}[H]
	\centering
	\footnotesize
  \caption{The Effects of Retirement on Self-Assessed Health Based on The Full Sample}
  \begin{minipage}{0.8\textwidth}
  \tabcolsep=0.1cm
\begin{tabular}{p{0.36\textwidth}>{\centering}p{0.3\textwidth}>{\centering\arraybackslash}p{0.3\textwidth}}

    \hline
    \hline
     & \textbf{1st Stage} & \textbf{2nd Stage}\\
     & Retirement & Self-Assessed Health\\     
    \midrule
    $\mathbf{1}$[Age$>$MRA] & 0.0462***  &   \\
       &  (0.0144)  &   \\  
  Retirement &  &   1.087** \\  
       &   &  (0.492)  \\    
\hline    
     Individual Fixed-Effect & Yes   & Yes  \\
    Observations & 93,818 & 93,818 \\
    Number of individuals  & 2,898 & 2,898 \\      
    \hline
    \end{tabular}
  \label{tab:SAH1}
    \scriptsize  \textit{Notes}: Self-assessed Health has been standardized. Standard errors clustered at individual are in parentheses, ***$p<0.01$, **$p<0.05$, * $p<0.1$. The second order polynomial function of age is controlled for all results.
    \end{minipage}
\end{table}%

\subsection{The Test of Justification Bias Based on Assumption 1}
\ \ \ \
In the previous subsection, we use the FE-IV method to rule out the effect of health on retirement. However, the resulting impact of retirement on self-assessed health can be due to either the shift in actual health or the reporting bias. In this subsection, we intend to separate these two by progressively narrowing the interval around the MRA. The first identification assumption we propose is that any sudden decline in subjective health immediately after retirement is likely to be driven by the reporting bias rather than the change in real health. Under this assumption, any effect estimated on a sufficiently small window surrounding the MRA will be considered as evidence for justification bias.

Table \ref{tab:main_as1} presents the two-stage least-squares (2SLS) results for the impact of retirement on self-assessed health at an interval of 10, 20, 30, 40, 50, or 60 months around the MRA.\footnote{The subsamples are symmetrically centered around the MRA. For instance, the result for the 60-month interval is derived from a period that spans 60 months both before and after the MRA.} The estimates reveal that retirement significantly affects self-assessed health in a window of 50 months and 60 months, similar to results based on the full sample. However, as the interval shrinks, the effect of retirement on self-assessed health remarkably diminishes and becomes statistically insignificant. By contrast, estimates in the first stage remain stable across intervals. These results demonstrate no immediate effect of retirement on self-assessed health and thus no evidence of the justification bias.

\begin{table}[H]
	\centering
	\footnotesize
  \caption{The Test of Justification Bias Based on Assumption 1}
  \begin{minipage}{0.99\textwidth}
  \tabcolsep=0.18cm
\begin{tabular}{p{0.24\textwidth}>{\centering}p{0.1\textwidth}>{\centering}p{0.1\textwidth}>{\centering}p{0.1\textwidth}>{\centering}p{0.1\textwidth}>{\centering}p{0.1\textwidth}>{\centering\arraybackslash}p{0.1\textwidth}}

    \hline
    \hline
     & 10 Months & 20 Months & 30 Months & 40 Months & 50 Months & 60 Months \\
    \midrule
      \multicolumn{7}{l}{\textbf{1st Stage}} \\   
  \qquad  $\mathbf{1}$[Age$>$MRA] & 0.0455*** & 0.0345*** & 0.0385*** & 0.0435*** & 0.0390*** & 0.0417***  \\
          & (0.0106) & (0.0125) & (0.0138) & (0.0145) & (0.0147) & (0.0147)  \\
   \multicolumn{3}{l}{\textbf{2nd Stage}}    &  &  &  &      \\
 \qquad    Retirement & 0.137 & 0.0719 & 0.0669 & 0.669 & 1.369** & 1.327**  \\
          & (0.424) & (0.472) & (0.408) & (0.436) & (0.679) & (0.624)  \\
          \hline
    Individual FE & Yes   & Yes   & Yes   & Yes   & Yes   & Yes   \\
    Observations & 12,040 & 23,044 & 33,462 & 42,942 & 51,426 & 59,321 \\
    Number of individuals & 1,044 & 1,299 & 1,554 & 1,735 & 1,984 & 2,177 \\
    \hline
    \end{tabular}
  \label{tab:main_as1}%

    \scriptsize \textit{Notes}: This table presents the estimated effects of retirement on self-Assessed health under various time frames. Self-assessed Health has been standardized. Standard errors clustered at individual are in parentheses, ***$p<0.01$, **$p<0.05$, * $p<0.1$. The second order polynomial function of age is controlled for all results.
    \end{minipage}
\end{table}%

However, the insignificance following interval reduction may result from imprecise estimates rather than a true effect of zero. As the interval narrows, the sample size decreases, which may lead to increased standard errors of estimators. Interestingly, the standard errors within narrower intervals become even smaller than the ones under wider intervals, which suggests that the statistically insignificant effects are mainly driven by the decrease in magnitude of coefficients rather than the increase in standard errors. Therefore, the null effect is unlikely to be driven by a lack of statistical power with a smaller sample size. 

To show the validity of our test, we further examine the estimates and standard errors under finer intervals. Figure \ref{fig:self-assessed Health (10-60 Months)} presents the results by reducing the interval every one month gradually. The estimates present a pattern similar to Table \ref{tab:main_as1}, with the effects approaching zero as the interval shrinks. Meanwhile, the standard errors remain stable regardless of the sample size.


\begin{figure}[H]
  \centering
\begin{minipage}{0.9\textwidth}
 \qquad \includegraphics[height=7cm]{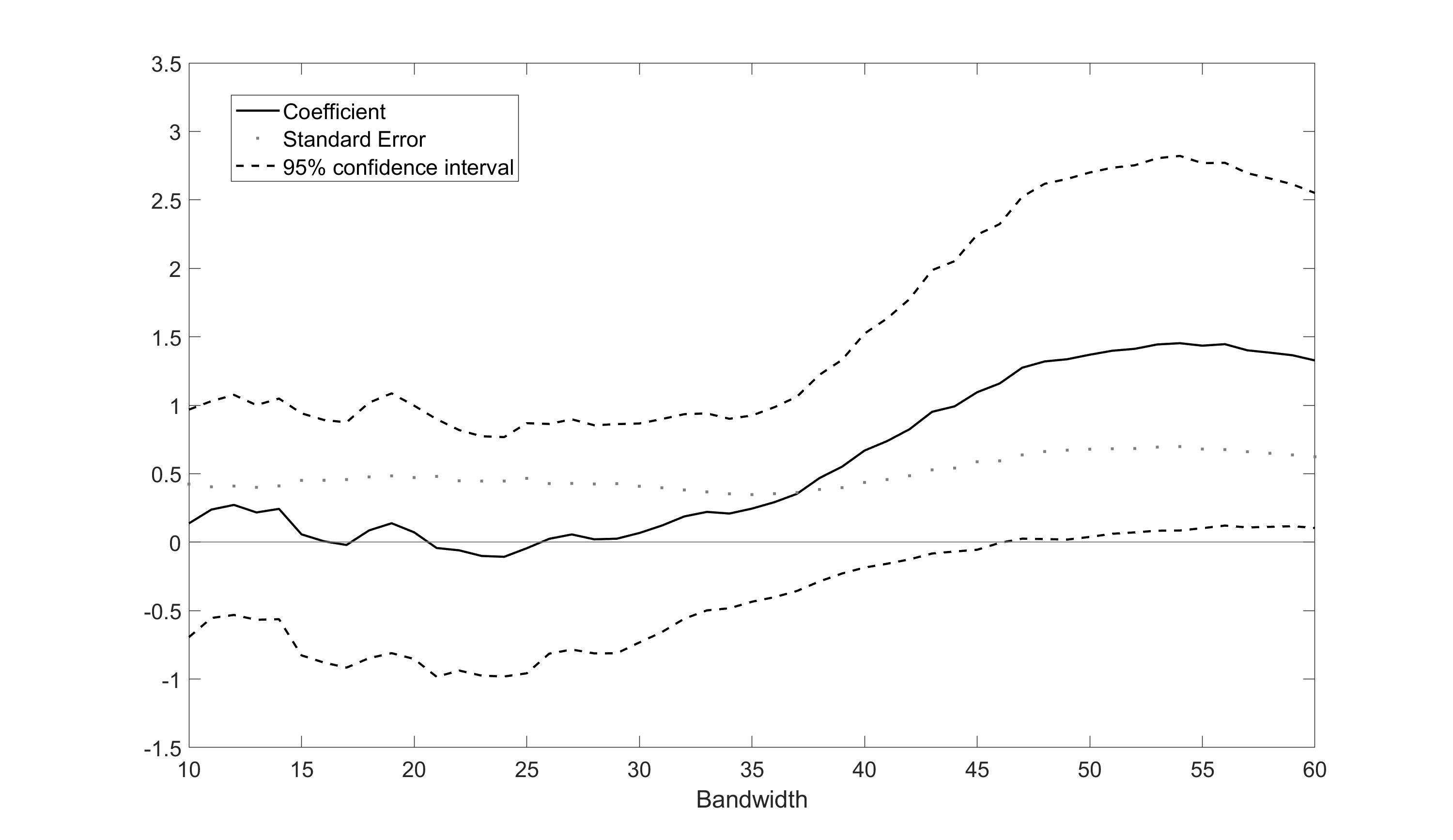}
  \caption{The 2SLS Coefficients and Standard Errors Under Finer Intervals} 
    \label{fig:self-assessed Health (10-60 Months)}
\vspace{0.2cm}
   \scriptsize \textit{Notes:} This figure depicts the estimated 2SLS coefficients and corresponding standard errors for the impact of retirement on self-assessed health under different intervals, from 10 to 60 months. The solid line represents the estimated coefficients, and the dash lines indicate the 95\% confidence intervals. The dot line represents standard errors of the coefficients.  
\end{minipage}%

\end{figure}

\subsection{The Test of Justification Bias Based on Assumption 2}
\ \ \ \ The preceding findings rely on the assumption that actual health does not undergo abrupt changes within a short time frame surrounding retirement.\footnote{Formally, the assumption is that conditional expectation function of the actual health is continuous at the MRA, as defined in the conceptual framework.} Any observed changes are therefore attributed to the reporting bias rather than the actual deterioration in health. While this assumption is plausible, it can be further relaxed by incorporating information on objective health into the analysis. 

Our second identification assumption and the corresponding Proposition 2 focus on the immediate effects of retirement on subjective relative to objective health, rather than solely on subjective health, as indicative of justification bias. Heuristically, our first identification assumption is similar to examine the variation only in the ``treatment group'', whereas the second assumption further uses the variation in objective health as the ``control group''. It thus exploits a design akin to the ``difference-in-differences'', i.e. the change in subjective health before and after the retirement versus the pre-post change in objective health.

To assess the impact of retirement on subjective relative to objective health, we create a new dependent variable, which characterizes the difference between the standardized subjective health and the standardized objective health. The resulting estimate is an indication of the extent to which subjective health exhibits extra change, in standard deviations, relative to objective health following retirement. We also provide an alternative specification, which is based on the regression similar to the one under Assumption 1, but additionally adding objective health conditions as control variables.

\begin{table}[H]
  \centering
  	\footnotesize
  \caption{The Test of Justification Bias Based on Assumption 2}
  \begin{minipage}{0.99\textwidth}
  \tabcolsep=0.18cm
\begin{tabular}{p{0.24\textwidth}>{\centering}p{0.1\textwidth}>{\centering}p{0.1\textwidth}>{\centering}p{0.1\textwidth}>{\centering}p{0.1\textwidth}>{\centering}p{0.1\textwidth}>{\centering\arraybackslash}p{0.1\textwidth}}
    \hline
    \hline
          & 10 Months & 20 Months & 30 Months & 40 Months & 50 Months & 60 Months \\
    \hline
  \multicolumn{7}{c}{ \textbf{Panel A: Difference in Subjective and Objective Health} }   \\
      \multicolumn{7}{l}{\textbf{1st Stage}} \\   
  \qquad  $\mathbf{1}$[Age$>$MRA] & 0.0455*** & 0.0345*** & 0.0385*** & 0.0435*** & 0.0390*** & 0.0417***  \\
          & (0.0106) & (0.0125) & (0.0138) & (0.0145) & (0.0147) & (0.0147)  \\
   \multicolumn{3}{l}{\textbf{2nd Stage}}    &  &  &  &      \\
    \qquad Retirement & 0.328 & -0.0398 & 0.239 & 1.033 & 1.767* & 1.570* \\
          & (0.721) & (0.737) & (0.634) & (0.664) & (0.932) & (0.819) \\
    \hline
  \multicolumn{7}{c}{  \textbf{Panel B: Objective Health as Controls} }\\
        \multicolumn{7}{l}{\textbf{1st Stage}} \\   
  \qquad  $\mathbf{1}$[Age$>$MRA] & 0.0433*** & 0.0337*** & 0.0382*** & 0.0432*** & 0.0389*** & 0.0415***  \\
          & (0.0106) & (0.0124) & (0.0138) & (0.0144) & (0.0147) & (0.0147)  \\
   \multicolumn{3}{l}{\textbf{2nd Stage}}    &  &  &  &      \\
  \qquad Retirement & 0.146 & 0.065 & 0.0789 & 0.693 & 1.395** & 1.344**  \\
          & (0.424) & (0.468) & (0.404) & (0.437) & (0.683) & (0.626) \\
    \hline
    Individual FE & Yes   & Yes   & Yes   & Yes   & Yes   & Yes \\
    Observations & 12,040 & 23,044 & 33,462 & 42,942 & 51,426 & 59,321 \\
    Number of individuals & 1,044 & 1,299 & 1,554 & 1,735 & 1,984 & 2,177  \\
    \hline
    \end{tabular}
      \label{tab:main_as2}%
\scriptsize  \textit{Notes:} Panel A represents the effect of retirement on the difference between subjective and objective health. The difference between subjective and objective Health is measured by the ``standardize self-assessed  health'' minus the ``standardize objective health''. Panel B represents the effect on subjective health by controlling objective health. Standard errors clustered at individual are in parentheses. ***$p<0.01$, **$p<0.05$, * $p<0.1$. The second order polynomial function of age is controlled for all results.    
     \end{minipage}
    \label{tab:Tests based on The Second Identification Assumption}%
\end{table}%

The results in Table \ref{tab:Tests based on The Second Identification Assumption} based on Assumption 2 show a pattern similar to our previous findings based on Assumption 1. Under wider intervals, retirement appears to have additional impacts on subjective health relative to objective health. However, as we restrict the sample to a narrower time frame, we no longer find any extra effects on subjective health relative to objective health. Again, we do not find evidences supporting the existence of justification bias.

\section{Additional Results}

\subsection{The Indirect Testing Method}
\ \ \ \ A popular approach to test justification bias indirectly is based on the regression of retirement on self-assessed health and compare the OLS and the IV estimate using objective health as instruments. Because both measurement errors and justification bias are arguably orthogonal to objective health conditions, the IV estimate is considered as unbiased. If justification bias dominates, it exaggerates the impact of health on retirement, and the OLS estimate should be larger than the IV counterpart.

\begin{table}[H]
  \centering
  \caption{The Indirect Method to Test Justification Bias}
  \begin{minipage}{0.8\textwidth}
    \footnotesize
    \begin{tabular}{p{0.39\textwidth}>{\centering}p{0.25\textwidth}>{\centering\arraybackslash}p{0.25\textwidth}}
    \hline
    \hline
              &\multicolumn{2}{c}{Retirement}  \\
          & OLS  & IV  \\
    \hline
    Self-Assessed Health & 0.017*** & 0.147*** \\
          & (0.005) & (0.039) \\
          \hline
    Individual FE & Yes   & Yes \\
    Observations & 72,183 & 72,183 \\
    Number of individuals & 2,360 & 2,360 \\
    \hline
    \end{tabular}

    \scriptsize \textit{Notes:} This table presents the OLS and the IV estimate of the effect of self-assessed health on retirement. Objective health conditions are used as instruments. In addition to individual fixed effects, we also control for non-parametric age functions (in years), dummy variables of the survey year, marital status, and the indicator of positive spousal income in the last month. Standard errors clustered at individual are in parentheses, ***$p<0.01$, **$p<0.05$, * $p<0.1$.
    \end{minipage}
  \label{Indirect Method to Test Justification Bias}%
\end{table}

Table \ref{Indirect Method to Test Justification Bias} compares the OLS and IV estimates of the effect of self-assessed health on retirement, where the IV estimate appears to be larger than the OLS estimate. This result suggests that the attenuation bias dominates, whereas the justification bias, if any, is not a serious problem when using self-assessed health as the measure. With conclusions consistent with this indirect approach, our new approach thus validates this popular conventional approach.

It is noteworthy that these two different testing approaches are implemented under the same sample, therefore the conclusions reached by them are directly comparable. Although previous studies typically find no justification bias by the above indirect method, whereas other methods find large justification bias, the conflicting findings may also arise from different social backgrounds and pension schemes across countries.

\subsection{Application to Standard Biennial Data Set}
\ \ \ \ To illustrate the advantage of the monthly data, we also apply our method to the HRS, a prototype of the extensive surveys about older people across many countries. The HRS sample used in the following analysis covers observations from 1998 (Wave 4) to 2018 (Wave 14). Detailed sample construction is illustrated in Appendix B.

The main reason that our method is also applicable to the biennial data set is that individuals' birth month is available. With variation in individual's birth month, we are able to construct the age variable also with monthly variation.\footnote{This is true if the date of interview is homogeneous. If the survey timing is heterogeneous, in general it introduces additional variation to the monthly age measure.} Nevertheless, even if birth months are uniformly distributed, observations that belong to a given age will be divided into 24 groups due to the biennial survey design.\footnote{Consider the example of a time frame encompassing observations of individuals aged from 60 to 62. In the HRS sample, at least 24 individuals within this age range, born in 24 different months, are needed to yield 24 observations. This is due to the biennial survey frequency. In contrast, leveraging the high-frequency sample, the same 24 individuals contribute to 24*24=596 observations because they are surveyed on a monthly basis.} Therefore the statistical power becomes a critical concern by applying our method to these typical data sets.

Figure \ref{fig:SE in HRS and SLP} compares the standard errors of the effect of retirement on self-assessed health using the HRS and the SLP sample, across different time frames. Due to insufficient sample size of the HRS sample, the first stage cannot be estimated when the interval is narrower than eight months around the statutory retirement age. Therefore, results are reported with an interval of eight months and more. The results indicate that when utilizing the HRS, the standard errors remain stable for intervals from 24 - 60 months, though with a gradual rise as the interval shrinks. With intervals narrower than that, the standard errors of coefficients start to surge, and the advantage of using the high-frequency data is evident. Our new method, however, has the potential to be applicable to those standard biennial data sets when more waves of data are accumulated.

\begin{figure}[H]
\centering
\begin{minipage}{0.8\textwidth}
\quad \includegraphics[height=7cm]{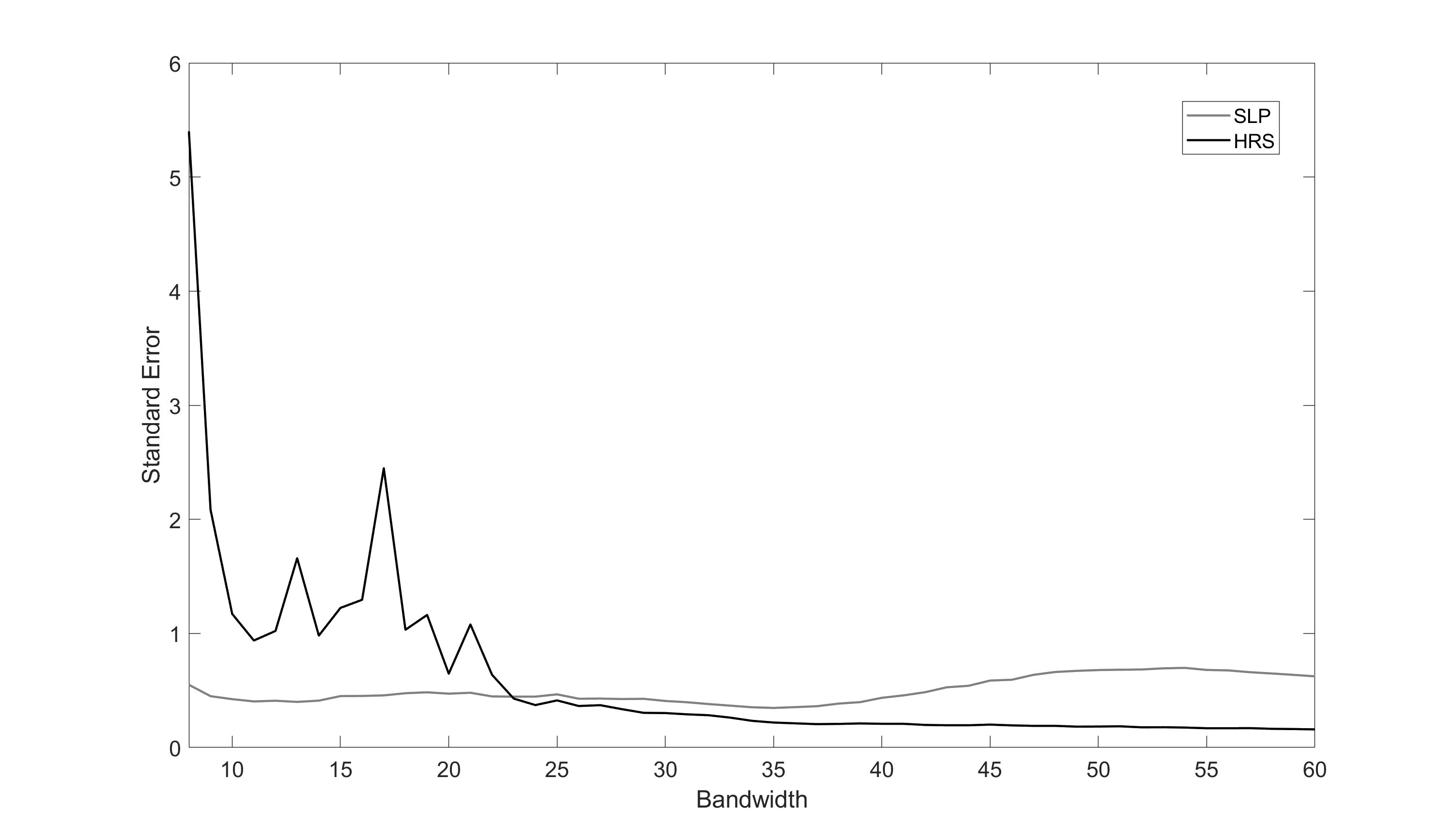}
  \caption{Standard Errors in HRS and SLP}
  \label{fig:SE in HRS and SLP}
  
  \vspace{0.2cm}
  \scriptsize \textit{Notes:} This graph presents standard errors using HRS and SLP data with different time frames, under the same estimation strategy. Early Retirement Age (age 62) is used as the instrument in the HRS result. The solid black line represents the standard error obtained using the HRS, while the solid gray line represents the result using the SLP.
\end{minipage}
\end{figure}

\subsection{Effects of Retirement on Objective Health}
\ \ \ \ Assumption 1 assumes that the actual health does not change suddenly after retirement. Although it is relaxed later by Assumption 2, Assumption 1 is testable with monthly data on objective health conditions, which are close proxy of the actual health.

\begin{table}[H]
  \centering
  \footnotesize
  \caption{The Effects of Retirement on Objective Health under Various Time Frames}
   \begin{minipage}{0.99\textwidth}
  \tabcolsep=0.18cm
\begin{tabular}{p{0.24\textwidth}>{\centering}p{0.1\textwidth}>{\centering}p{0.1\textwidth}>{\centering}p{0.1\textwidth}>{\centering}p{0.1\textwidth}>{\centering}p{0.1\textwidth}>{\centering\arraybackslash}p{0.1\textwidth}}
    \hline
    \hline
          & 10 Months & 20 Months & 30 Months & 40 Months & 50 Months & 60 Months \\
    \hline
    \textbf{2nd Stage} & & & & & & \\
   \qquad Retirement & -0.128 & 0.0749 & -0.115 & -0.244 & -0.267 & -0.163 \\
          & (0.403) & (0.424) & (0.365) & (0.333) & (0.382) & (0.348)\\
    \hline
    Individual Fixed-Effect & Yes   & Yes   & Yes   & Yes   & Yes   & Yes  \\
    Observations & 12,040 & 23,044 & 33,462 & 42,942 & 51,426 & 59,321\\
    Number of individuals & 1,044 & 1,299 & 1,554 & 1,735 & 1,984 & 2,177 \\
    \hline
    \end{tabular}%
  \label{tab:Objective Health}

     \scriptsize \textit{Notes:} This table presents the effect of retirement on objective health under various time frames around the MRA. Objective health is measured by the sum of indicators of various objective health conditions. Standard errors clustered at individual are in parentheses, ***$p<0.01$, **$p<0.05$, * $p<0.1$. The second order polynomial function of age is controlled for all results.
    \end{minipage}

\end{table}%

Table \ref{tab:Objective Health} tests Assumption 1 by examining how retirement affects individuals' objective health under various time frames around the MRA. The results show that retirement has no significant impact on objective health regardless of the width of interval.



\section{Robustness Checks}

\subsection{Adding Additional Control Variables}
\ \ \ \ 
Similar to the regression discontinuity design, the new testing method should hold irrespective of the control variables, because the assumption is that observations right before and after the cutoff should be similar and comparable. Moreover, controlling for individual fixed effects further eliminates a number of omitted variables that are time invariant, such as education and gender. 

\begin{table}[H]
  \centering
  \footnotesize
  \caption{Results with Additional Time-Varying Control Variables}
  \begin{minipage}{0.99\textwidth}
  \tabcolsep=0.19cm
\begin{tabular}{p{0.23\textwidth}>{\centering}p{0.1\textwidth}>{\centering}p{0.1\textwidth}>{\centering}p{0.1\textwidth}>{\centering}p{0.1\textwidth}>{\centering}p{0.1\textwidth}>{\centering\arraybackslash}p{0.1\textwidth}}
    \hline
    \hline
    & 10 Months & 20 Months & 30 Months & 40 Months & 50 Months & 60 Months \\
    \hline
    \multicolumn{7}{l}{\textbf{1st Stage}} \\   
   \qquad $\mathbf{1}$[Age$>$MRA]  & 0.0447*** & 0.0355*** & 0.0396*** & 0.0453*** & 0.0414*** & 0.0437*** \\
          & (0.0108) & (0.0126) & (0.0139) & (0.0146) & (0.0149) & (0.0149) \\
    \multicolumn{3}{l}{\textbf{2nd Stage}}    &  &  &  &      \\
         
   \qquad Retirement & -0.0386 & 0.0421 & 0.00161 & 0.592 & 1.192** & 1.210**  \\
          & (0.420) & (0.454) & (0.396) & (0.406) & (0.594) & (0.567) \\
    \midrule
    \multicolumn{1}{l}{Individual FE} & Yes   & Yes   & Yes   & Yes   & Yes   & Yes  \\
    Observations & 11,787 & 22,539 & 32,677 & 41,930 & 50,185 & 57,852 \\
    Number of Individuals & 999   & 1,255 & 1,487 & 1,683 & 1,925 & 2,128  \\
    \hline
    \end{tabular}%
  \label{tab:self-assessed Health Effect with More Control Variables}%

 \scriptsize \textit{Notes:} This table presents the results controlling for additional time-varying variables, including the dummies of the survey year, marital status, and the dummy of having positive expense on health insurance. Standard errors clustered at individual are in parentheses, ***$p<0.01$, **$p<0.05$, * $p<0.1$. The second order polynomial function of age is controlled for all results.

      \end{minipage}
\end{table}%
Nevertheless, we conduct a robustness check by adding additional time-varying variables as controls. In this new specification, we control for dummies of the survey years, individuals' marital status, as well as a dummy variable indicating any positive expenditure on health insurance. The results remain robust as shown in Table \ref{tab:self-assessed Health Effect with More Control Variables}.

\subsection{Removing Individual Fixed Effects}
\ \ \ \ Similarly, if the new testing approach is valid, including individual fixed effects should not have great impact on our estimates, although it may affect standard errors of the estimates as we are exploiting the within-individual variation to identify the model \citep{lee2010regression}. To check the sensitivity of our results, we replicate our main analysis after removing the individual fixed effects. As shown in Table \ref{tab:self-assessed Health Effect without Fixed Effect}, we still obtain similar results, although the standard errors become slightly larger.

\begin{table}[H]
  \centering
  \footnotesize
  \caption{Results without Individual Fixed Effects}
   \begin{minipage}{0.99\textwidth}
  \tabcolsep=0.19cm
\begin{tabular}{p{0.23\textwidth}>{\centering}p{0.1\textwidth}>{\centering}p{0.1\textwidth}>{\centering}p{0.1\textwidth}>{\centering}p{0.1\textwidth}>{\centering}p{0.1\textwidth}>{\centering\arraybackslash}p{0.1\textwidth}}
    \hline
    \hline
       & 10 Months & 20 Months & 30 Months & 40 Months & 50 Months & 60 Months \\
    \hline
      \multicolumn{7}{l}{\textbf{1st Stage}} \\   
    \qquad $\mathbf{1}$[Age$>$MRA]  & 0.0439*** & 0.0343*** & 0.0394*** & 0.0435*** & 0.0391*** & 0.0417*** \\
          & (0.0106) & (0.0124) & (0.0137) & (0.0144) & (0.0147) & (0.0147) \\
  \multicolumn{7}{l}{\textbf{2nd Stage}} \\
   \qquad Retirement & 0.136 & 0.0691 & 0.0681 & 0.665 & 1.368** & 1.315** \\
          & (0.424) & (0.470) & (0.407) & (0.433) & (0.678) & (0.619)  \\
    \hline
    \multicolumn{1}{l}{Individual FE} & No    & No    & No    & No    & No    & No     \\
    Observations & 12,040 & 23,044 & 33,462 & 42,942 & 51,426 & 59,321 \\
    Number of Individuals & 1,044 & 1,299 & 1,554 & 1,735 & 1,984 & 2,177 \\
    \hline
    \end{tabular}%
  \label{tab:self-assessed Health Effect without Fixed Effect}%

    \scriptsize   \textit{Notes: } This table presents the results after removing the individual fixed-effects. Standard errors clustered at individual are in parentheses, ***$p<0.01$, **$p<0.05$, * $p<0.1$. The second order polynomial function of age is controlled for all results.
      \end{minipage}
    
\end{table}%

\subsection{Different Functional Forms of Age Controls}
\ \ \ \
Following the literature that estimates the effect of retirement on health by FE-IV, our baseline specifications control for a quadratic function of age to absorb the smooth effect of aging on retirement and health. In this subsection, we allow for alternative functional forms to check the sensitivity of our results. 

Health may evolve differently after the retirement if individuals' habits have changed, which may be reflected in a shifting slope in the conditional expectation function of actual health. Therefore, allowing for changes in the slopes of age functions after reaching the MRA in both the retirement equation \eqref{eqIV1} and the health equation \eqref{eqIV2}, whereby $g_R(X_{it})= \kappa_{R,1} X_{it} + \kappa_{R,2} X_{it} \cdot  \mathbf{1}\{X_{it}>m^*_{i} \}$ and $g_H(X_{it})= \kappa_{H,1} X_{it} + \kappa_{H,2} X_{it} \cdot  \mathbf{1}\{X_{it}>m^*_{i} \}$, we obtain Equation \eqref{robust1} and \eqref{robust2}:
\begin{equation}
R_{it}=\alpha+\rho \cdot \mathbf{1}\{X_{it}>m^*_{i} \}+\kappa_{R,1} X_{it} + \kappa_{R,2} X_{it} \cdot  \mathbf{1}\{X_{it}>m^*_{i} \}+\phi_R^{'}W_{it}+\theta_i+e_{it}
 \label{robust1}
\end{equation}
\begin{equation} 
H_{it}=\beta+\gamma \hat{R}_{it}+\kappa_{H,1} X_{it} + \kappa_{H,2} X_{it} \cdot  \mathbf{1}\{X_{it}>m^*_{i} \}+ \phi_H^{'}W_{it}+\sigma_{i}+\epsilon_{it}
 \label{robust2}
\end{equation}

Table \ref{tab:other specifications} present the estimation results, which are similar to the main results in Table \ref{tab:main_as1}. We still observe a detrimental impact of retirement on self-assessed health that diminishes as the interval narrows, whereas the standard errors remain stable. Meanwhile, the effect becomes statistically insignificant when the window is small enough.\footnote{We have also tried more flexible functional forms and the coefficients remain insignificant. However, the estimates become very imprecise, possibly due to multicollinearity. We are thus less confident that those results are solid evidences of an absence of justification bias.}

\begin{table}[H]
  \centering
  \footnotesize
  \caption{Different Functional Forms of Age Controls}
  \begin{minipage}{0.99\textwidth}
  \tabcolsep=0.18cm
\begin{tabular}{p{0.24\textwidth}>{\centering}p{0.1\textwidth}>{\centering}p{0.1\textwidth}>{\centering}p{0.1\textwidth}>{\centering}p{0.1\textwidth}>{\centering}p{0.1\textwidth}>{\centering\arraybackslash}p{0.1\textwidth}}
    \hline
    \hline
          & 10 Months & 20 Months & 30 Months & 40 Months & 50 Months & 60 Months \\
    \midrule
      \multicolumn{7}{l}{\textbf{1st Stage}} \\   
   \qquad  $\mathbf{1}$[Age$>$MRA] & 0.0455*** & 0.0345*** & 0.0385*** & 0.0435*** & 0.0390*** & 0.0417*** \\
          & (0.0106) & (0.0125) & (0.0138) & (0.0145) & (0.0147) & (0.0147) \\
      \multicolumn{7}{l}{\textbf{2nd Stage}} \\ 
   \qquad Retirement & 0.100   & 0.0745 & 0.0713 & 0.666 & 1.364** & 1.321** \\
          & (0.407) & (0.464) & (0.406) & (0.434) & (0.675) & (0.621) \\
    \hline
    \multicolumn{1}{l}{Individual FE} & Yes   & Yes   & Yes   & Yes   & Yes   & Yes  \\
    Observations & 12,040 & 23,044 & 33,462 & 42,942 & 51,426 & 59,321 \\
    Number of Individuals & 1,044 & 1,299 & 1,554 & 1,735 & 1,984 & 2,177 \\
\hline
    \end{tabular}%
  \label{tab:other specifications}%
  
 \scriptsize \textit{Notes:} This table presents the results based on age functions that have different slopes before and after the MRA. Standard errors clustered at individual are in parentheses, ***$p<0.01$, **$p<0.05$, * $p<0.1$. The second order polynomial function of age is controlled for all results.
     \end{minipage}
\end{table}%

\subsection{Using Alternative Dependent Variable}
\ \ \ \ 
The main analysis is based on the raw measure of self-assessed health with a five-point scale. In this subsection, we consider an alternative dependent variable, the indicator of poor health, which takes a value of one if respondents answered ``fair health'' or ``poor health'' and zero otherwise.  As shown in Table \ref{tab:poor subjective health effect}, we still observe that when the interval around the MRA is wide, retirement has a significant detrimental effect. However, as the time frame narrows, the coefficients become small and insignificant. 
\begin{table}[H]
	\centering
	\footnotesize
  \caption{Results Using The Alternative Dependent Variable}
  \begin{minipage}{0.99\textwidth}
  \tabcolsep=0.19cm
\begin{tabular}{p{0.23\textwidth}>{\centering}p{0.1\textwidth}>{\centering}p{0.1\textwidth}>{\centering}p{0.1\textwidth}>{\centering}p{0.1\textwidth}>{\centering}p{0.1\textwidth}>{\centering\arraybackslash}p{0.1\textwidth}}

    \hline
    \hline
     & 10 Months & 20 Months & 30 Months & 40 Months & 50 Months & 60 Months \\
    \midrule
      \multicolumn{7}{l}{\textbf{1st Stage}} \\   
  \qquad  $\mathbf{1}$[Age$>$MRA] & 0.0455*** & 0.0345*** & 0.0385*** & 0.0435*** & 0.0390*** & 0.0417***  \\
          & (0.0106) & (0.0125) & (0.0138) & (0.0145) & (0.0147) & (0.0147)  \\
   \multicolumn{3}{l}{\textbf{2nd Stage}}    &  &  &  &      \\
 \qquad    Retirement  & 0.0493 & 0.186 & 0.116 & 0.291 & 0.571* & 0.562*\\
         & (0.235) & (0.278) & (0.231) & (0.229) & (0.328) & (0.304)  \\
          \hline
    Individual FE & Yes   & Yes   & Yes   & Yes   & Yes   & Yes   \\
    Observations & 12,040 & 23,044 & 33,462 & 42,942 & 51,426 & 59,321 \\
    Number of individuals & 1,044 & 1,299 & 1,554 & 1,735 & 1,984 & 2,177 \\
    \hline
    \end{tabular}
  \label{tab:poor subjective health effect}%

    \scriptsize \textit{Notes}: This table presents the results using the indicator of poor health as the dependent variable. Standard errors clustered at individual are in parentheses, ***$p<0.01$, **$p<0.05$, * $p<0.1$. The second order polynomial function of age is controlled for all results.
    \end{minipage}
\end{table}%


\section{Conclusion}
 \ \ \ \
Population aging urges many governments to reform their retirement policies. To ensure the efficacy of such policy reforms, it is important to understand the role of health in retirement, as poor health may inhibit older workers' responsiveness to these policy changes.  While most researchers agree that self-assessed health is an important measure for an individual's actual health, justification bias is a major concern. This paper proposes a new approach to test this bias with two alternative identification assumptions. Exploiting a unique high-frequency data that provides monthly information on employment, self-assessed health, as well as objective health conditions, we find no evidence of justification bias in the context of Singapore. 

The high-frequency data is critical to ensure sufficient statistical power to apply our new method, especially when the null hypothesis of having no bias cannot be rejected. However, our method is generally applicable to the standard biennial survey data sets, such as HRS, SHARE, CHARLS, ELSA etc, as long as detailed birth information is available. While the high-frequency monthly data is indispensable at this stage, we expect our approach offers another possibility to test justification bias as more data are available in the future. 

\newpage

\bibliographystyle{apalike}
\bibliography{reference}

\newpage

\appendix


\section*{Appendix A: Proof of Proposition 3}
\ \ \ \ To see Proposition 3 holds, we focus on the identification strategy under Assumption 1 and consider a small time frame around the Minimum Retirement Age(MRA), one of the statutory retirement ages in Singapore.  For observations within this interval, by the retirement responsiveness to extra incentives generated by the statutory age, there are three types of individuals: never-takers who did not retire; always-takers who already retired; and compliers who retire at the statutory age, $  m^{*}\equiv \text{MRA}$.\footnote{This classification is asymptotically true as the interval shortens and converges to the MRA.} Therefore the conditional expectation function of self-assessed health can be defined respectively for these three types of individuals:
\begin{itemize}
    \item Never-takers: $ E[H^s| X, R, Type= N] = E[H^s| X, R=0] = f_0 (X)   $
    \item Always-takers: $E[H^s| X, R, Type= A] = E[H^s| X, R=1] = \tilde{f_1} (X)   $
    \item Compliers:  $ E[H^s| X, R, Type= C] =  E[H^s| X, R] =  (1-R)\cdot f_0(X) +  R\cdot \tilde{f_1}(X) $
\end{itemize}
The observed health is an average over the never-takers, always-takers, and compliers, with corresponding probability mass $P^N, P^A, \text{ and } p^C$. Then we can prove that the justification bias is still identified by the abrupt change of self-assessed health around the MRA.

\begin{proof}
\begin{align*}
& \lim_{\delta \rightarrow 0 } E[H^s | X=m^*+\delta, R]  -  \lim_{\delta \rightarrow 0 } E[H^s | X=m^*-\delta, R]    \\
= &   \lim_{\delta \rightarrow 0 } \big \{   E[H^s|  X=m^*+\delta, R, T=N] \cdot P^{N}  + E[H^s|  X=m^*+\delta, R,T= A] \cdot P^{A}\\
 & \qquad +  E[H^s|  X=m^*+\delta, R, T= C] \cdot P^{C} \big \} \\
 &  - \lim_{\delta \rightarrow 0 } \big \{   E[H^s|  X=m^*-\delta, R,T=  N] \cdot P^{N} + E[H^s|  X=m^*-\delta, R, T= A] \cdot P^{A} \\ 
 & \qquad  +  E[H^s|  X=m^*-\delta, R, T= C] \cdot P^{C} \big \} \\
 = &  \lim_{\delta \rightarrow 0 } \big \{ f_0(m^* +\delta) - f_0(m^* -\delta) \big \} \cdot P^N +  \lim_{\delta \rightarrow 0 } \big \{ \tilde{f_1} (m^* +\delta)  - \tilde{f_1} (m^* -\delta) \big \}   \cdot  P^A \\
 & +  \lim_{\delta \rightarrow 0 }  \big \{   \tilde{f_1} (m^* +\delta)  -   f_0 (m^*-\delta)   \} \cdot  P^C  \big \} \\
 =&   \quad   \lambda \cdot P^C  &&\qedhere
\end{align*}
\end{proof}
That is, as long as $P^C>0$, the justification bias is identified by:
\begin{align*}
 \lambda= \frac  { \lim_{\delta \rightarrow 0 } E[H^s | X=m^*+\delta, R]  -  \lim_{\delta \rightarrow 0 } E[H^s | X=m^*-\delta, R] } { P^C}
\end{align*}

\section*{Appendix B: The HRS Sample }
\ \ \ \ For our HRS sample, the choice of starting from Wave 4 is due to the inconsistent entry times of the HRS cohort and AHEAD cohort in the first three survey waves. Additionally, to maintain consistency in sample selection with the baseline of this study and to exclude the impact of the COVID-19 pandemic, samples conducted in the year 2020 (Wave 15) were excluded.

In the HRS database, the question regarding the respondents' employment status is measured using ``current job status.'' Job status options include ``working now'', ``temporarily laid off'', ``unemployed and looking for work'', ``disabled'', ``retired'', ``homemaker'', and ``other''. Respondents are allowed to choose all options that apply. The measure of self-assessed health in HRS is the same as that in SLP. To maintain consistency with our baseline, when selecting samples for the HRS, we first exclude those samples that had selected multiple options for ``current job status'' and drop samples that answered "homemaker" and "other." Additionally, we follow the same sample processing as the baseline, excluding samples where the employment status remained unchanged. Finally, we just keep samples of age from 55 to 75. After the sample restriction, we ultimately obtained 52,679 observations. 


\begin{table}[H]
	\centering
	\footnotesize
  \caption{Effects of Retirement on Self-Assessed Health by HRS Sample}
  \begin{minipage}{0.99\textwidth}
  \tabcolsep=0.19cm
\begin{tabular}{p{0.23\textwidth}>{\centering}p{0.1\textwidth}>{\centering}p{0.1\textwidth}>{\centering}p{0.1\textwidth}>{\centering}p{0.1\textwidth}>{\centering}p{0.1\textwidth}>{\centering\arraybackslash}p{0.1\textwidth}}

    \hline
    \hline
     & 10 Months & 20 Months & 30 Months & 40 Months & 50 Months & 60 Months \\
    \hline
      \multicolumn{7}{l}{\textbf{1st Stage}} \\   
  \qquad  $\mathbf{1}$[Age$>$62] & 0.436 & 0.0978** & 0.0684*** & 0.0811*** & 0.0828*** & 0.0858*** \\
   & (0.404) & (0.0402) & (0.0124) & (0.0102) & (0.00948) & (0.00871) \\
   \multicolumn{3}{l}{\textbf{2nd Stage}}    &  &  &  &      \\
 \qquad    Retirement & 0.164 & 0.453 & -0.165 & 0.103 & 0.0844 & 0.140 \\
          & (1.171) & (0.647) & (0.302) & (0.208) & (0.184) & (0.159) \\
          \hline
    Individual FE & Yes   & Yes   & Yes   & Yes   & Yes   & Yes   \\
    Observations & 5,482 & 10,541 & 15,645 & 20,487 & 25,202 & 29,685 \\
    Number of individualsn & 5,380 & 6,996 & 7,669 & 8,108 & 8,461 & 8,709 \\
    \hline
    \end{tabular}
  \label{tab:HRSresults}%

    \scriptsize \textit{Notes}: Self-assessed health has been standardized. Standard errors clustered at individual are in parentheses, ***$p<0.01$, **$p<0.05$, * $p<0.1$. The second order polynomial function of age is controlled for all results. This estimation is made by HRS sample.
    \end{minipage}
\end{table}%

\end{document}